\documentclass[12pt,doublespace,onecolumn]{IEEEtran}

\usepackage{graphicx}
\usepackage{epstopdf}
\usepackage{setspace}
\usepackage{amsmath}

\ifCLASSINFOpdf

\else

\fi
\usepackage{amsmath,epsfig,amssymb,verbatim,amsopn,cite,subfigure,multirow}
\usepackage{balance}
\usepackage[usenames,dvipsnames]{color}
\usepackage[all]{xy}  
\usepackage{url}
\usepackage{amsfonts}
\usepackage{amssymb}
\usepackage{epsfig}
\usepackage{slashbox}
\usepackage{bm}
\usepackage{epsfig}
\usepackage{amsmath,amssymb, amstext}
\usepackage{graphicx}
\usepackage{epstopdf}
\newtheorem{theo}{Theorem}
\newtheorem{exmp}{Example}
\newtheorem{step}{Step}[exmp]
\newtheorem{iter}{Iteration}
\newtheorem{resu}{Result}
\usepackage[T1]{fontenc}
\usepackage[table]{xcolor}

\begin{document}


\title{Parity Check Matrix Recognition from Noisy Codewords}


\author{\IEEEauthorblockN{Yasser Karimian, Saeideh Ziapour and Mahmoud Ahmadian Attari }
\IEEEauthorblockA{ \\Coding and Cryptography Lab. Faculty of Electrical and Computer Engineering\\ K. N. Toosi University of Technology\\
Tehran, Iran\\
yasser.karimian@ee.kntu.ac.ir, mahmoud@eetd.kntu.ac.ir}}

\begin{spacing}{2}

\maketitle
\begin{abstract}
We study recovering parity check relations for an unknown code from intercepted bitstream received from Binary Symmetric Channel in this paper. An iterative column elimination algorithm is introduced which attempts to eliminate parity bits in codewords of noisy data. This algorithm is very practical due to low complexity and use of XOR operator. Since, the computational complexity is low, searching for the length of code and synchronization is possible. Furthermore, the Hamming weight of the parity check words are only used in threshold computation and unlike other algorithms, they have negligible effect in the proposed algorithm. Eventually, experimental results are presented and estimations for the maximum noise level allowed for recovering the words of the parity check matrix are investigated.
\end{abstract}

\begin{IEEEkeywords}
linear block code, parity check matrix, iterative column elimination, parity elimination.
\end{IEEEkeywords}

\newpage

\section{Introduction}
\label{Sec: Intro}
Error correction codes are used in telecommunication systems to deal with noise and to increase data transmission ability. The message is encoded at the transmitter by channel encoder and decoded at the receiver knowing the parameters of code such as generator and parity check matrices. Specification recovery for communication systems from received signal, without any knowledge about the transmitter system is very complicated. The solution in \cite{Valembois2001199} is to find the nearest (in sense of Hamming distance) $(n,k)$-code from the code used for channel coding. The associated decision problem is proved to be NP-complete. \cite{Cluzeau:isit06} introduces an iterative decoding based algorithm using Gallager decoder which seeks for low-weight words. \cite{Cluzeau:isit06} and \cite{Cluzeau:isit08} address the problem as reverse-engineering a communication systems.

In \cite{Barbier:amcs06} and \cite{Barbier:phd07} an algorithm based on rank computation has been introduced to find the correct length and synchronization, but a simple method is needed to find the ranks of sub-matrices. Two algorithms were proposed in \cite{Cluzeau:isit09} to search for words of parity check matrix which can obtain code length and synchronization. If selected code length and synchronization are chosen correctly, looking for dual codewords will be sufficient for obtaining information of code length, synchronization, and words in the dual code \cite{Cluzeau:isit09}. To decrease search space, Canteaut-Chabaud information set decoding algorithm \cite{canteant:transinfo98} has been used. However, we need a proper choice on the weight of parity check matrix words to recover them.

In this paper we present a very low complexity algorithm to find the words of a parity check matrix without any assumption about the weight of the words. The main idea of this work is to use the parity check equations which lead to zero syndrome bits. In fact, to recover any word of the parity check matrix, $h$, we only need $k$ independent codewords ${\upsilon _j},\;1 \le j \le k$, such that they satisfy the parity equation ${\upsilon _j}{h^T} = 0$.

The paper is organized as follows: In Section II column elimination operation and linear block codes are discussed. Parity elimination on linear block codes is presented in section III. In the next section iterative elimination algorithm is introduced for recovering dual codewords from noisy data. The required equations for threshold computation are obtained in section V. Experimental results are given in section VI, and finally section VII concludes the paper.

\section{Preliminaries}
\subsection{Column Elimination Operation}
Consider matrix $\Lambda $,
\begin{equation}
 {\Lambda }=\left(\begin{IEEEeqnarraybox}[
\IEEEeqnarraystrutmode
][c]{c,c,c}
a_{11} &\cdots & a_{1n}\\
\vdots & \ddots  &  \vdots \\
a_{M1} & \cdots   &  a_{Mn}
\end{IEEEeqnarraybox}\right)
\end{equation}
where ${a_{ij}} = \left\{ {0,1} \right\}$ for $1 \le i \le n,1 \le j \le M$. Linear independent columns of matrix $\Lambda$
form a basis for vector space $\textbf{\emph{S}}$, and every column of $\Lambda$ is an element of $\textbf{\emph{S}}$. Elementary column
operation can be used to obtain column echelon form in order to find basis set, and it is
performed on $\Lambda$ as follows:

\textit{Step 1}: If ${a_{11}} \ne 0$ first colmn of $\Lambda$ is assumed as a basis vector of vector space $\textbf{\emph{S}}$ and is eliminated from
other dependent vectors by multiplication of transition matrix ${\Gamma} ^{(1)}$. Therefore, ${{\Lambda} ^{(1)}}$, the changed
matrix in the first step, is
\begin{equation}
{\rm{\Lambda }}^{(1)} = {\rm{\Lambda }} \times {\Gamma ^{(1)}} = \left( {\begin{array}{*{20}{c}}
{{a_{11}}}& \cdots &{{a_{1n}}}\\
 \vdots & \ddots & \vdots \\
{{a_{M1}}}& \cdots &{{a_{Mn}}}
\end{array}} \right)
 \left(\begin{IEEEeqnarraybox}[
\IEEEeqnarraystrutmode
\IEEEeqnarraystrutmode
][c]{c/v/c,c,c}
a_{11} && a_{12}  & \cdots & a_{12}
\\\hline
0 &&  &  &   \\
\vdots &&   &  I_{n} &   \\
0 &&   &  &
\end{IEEEeqnarraybox}\right)=
\left( {\begin{array}{*{20}{c}}
1&0& \cdots &0\\
{a_{21}^{(1)}}&{a_{22}^{(1)}}& \cdots &{a_{2n}^{(1)}}\\
 \vdots & \vdots & \ddots & \vdots \\
{a_{M1}^{(1)}}&{a_{M2}^{(1)}}& \cdots &{a_{Mn}^{(1)}}
\end{array}} \right)
\end{equation}
If $a_{11}=0$, the first row should be replaced by a row of non-zero first element, then the primary column
operation can be performed.

\textit{Step 2}: If $a_{22}=0$, multiplication of transition matrix, $\Gamma ^{(2)}$, results in

\begin{equation}
{\rm{\Lambda }}^{(2)} = {\rm{\Lambda }}^{(1)} \times {\Gamma ^{(2)}} = \left( {\begin{array}{*{20}{c}}
1&0& \cdots &0\\
{a_{21}^{(1)}}&{a_{22}^{(1)}}& \cdots &{a_{2n}^{(1)}}\\
 \vdots & \vdots & \ddots & \vdots \\
{a_{M1}^{(1)}}&{a_{M2}^{(1)}}& \cdots &{a_{Mn}^{(1)}}
\end{array}} \right)
 \left(\begin{IEEEeqnarraybox}[
\IEEEeqnarraystrutmode
][c]{c,c/v/c,c,c}
1 & 0 &&  & \cdots & 0  \\
0 & a_{22}^{(1)}  &&  & \cdots & a_{2n}^{(1)}
\\\hline
0 & 0 &&  & &  \\
\vdots & \vdots  &&  &I_{n} &   \\
0 & 0  &&  & &
\end{IEEEeqnarraybox}\right)=
\left( {\begin{array}{*{20}{c}}
{\begin{array}{*{20}{c}}
1&0&0\\
{a_{21}^{(2)}}&1&0\\
{a_{31}^{(2)}}&{a_{32}^{(2)}}&{a_{33}^{(2)}}
\end{array}\begin{array}{*{20}{c}}
 \cdots &0\\
 \cdots &0\\
 \cdots &{a_{3n}^{(2)}}
\end{array}}\\
{\begin{array}{*{20}{c}}
 \vdots & \vdots & \vdots \\
{a_{M1}^{(2)}}&{a_{M2}^{(2)}}&{a_{M3}^{(2)}}
\end{array}\begin{array}{*{20}{c}}
 \ddots & \vdots \\
 \cdots &{a_{Mn}^{(2)}}
\end{array}}
\end{array}} \right)
\end{equation}
Similarly, if $a_{22}=0$, the second row should be replaced by a row of non-zero second element, then the primary column operation can be performed.

\textit{Step j}: If $a_{jj} \ne 0$, multiplication of transition matrix, $\Gamma ^{(j)}$, results in
\begin{align}
{\rm{\Lambda }}^{(j)} = {\rm{\Lambda }}^{(j-1)} \times {\Gamma ^{(j)}} &=
\left(\begin{IEEEeqnarraybox}[
\IEEEeqnarraystrutmode
][c]{c,c,c,c,c,c}
1&0&0&0& \cdots &0\\
 \vdots & \ddots &{}&{}&{}& \vdots \\
{a_{\left( {j - 1} \right)1}^{\left( {j - 1} \right)}}& \cdots &1&0& \cdots &0\\
{a_{j1}^{(j - 1)}}& \cdots &{a_{j(j - 1)}^{(j - 1)}}&{a_{jj}^{(j - 1)}}& \cdots &{a_{jn}^{(j - 1)}}\\
 \vdots & \ddots & \vdots & \vdots & \ddots & \ddots \\
{a_{M1}^{(j - 1)}}& \cdots &{a_{Mj}^{(j - 1)}}&{a_{Mj}^{(j - 1)}}& \cdots &{a_{Mn}^{(j - 1)}}
\end{IEEEeqnarraybox}\right)
 \left(\begin{IEEEeqnarraybox}[
\IEEEeqnarraystrutmode
][c]{c,c,c/v/c,c,c}
  &         &  && 0 & \cdots & 0  \\
  & I_{j-1} &  &&   &        & \vdots  \\
  &         &  && 0 & \cdots & 0
  \\\hline
0 & \cdots  &0 && a_{jj}^{\left( {j - 1} \right)}  &\cdots& a_{jn}^{\left( {j - 1} \right)}
\\\hline
 0&  \cdots &0 &&   &       &     \\
\vdots & \ddots  & \vdots &&   & I_{n-j} &   \\
0 & \cdots  &0 &&   &     &
\end{IEEEeqnarraybox}\right)\nonumber\\
&=\left(\begin{IEEEeqnarraybox}[
\IEEEeqnarraystrutmode
][c]{c,c,c,c,c,c}
1&0&0&0& \cdots &0\\
 \vdots & \ddots &{}&{}&{}& \vdots \\
{a_{j1}^{(j)}}& \cdots &1&0& \cdots &0\\
{a_{(j + 1)1}^{(j)}}& \cdots &{a_{(j + 1)j}^{(j)}}&{a_{\left( {j + 1} \right)(j + 1)}^{(j)}}& \cdots &{a_{(j + 1)n}^{(j)}}\\
 \vdots & \ddots & \vdots & \vdots & \ddots & \ddots \\
{a_{M1}^{(j)}}& \cdots &{a_{Mj}^{(j)}}&{a_{M(j + 1)}^{(j)}}& \cdots &{a_{Mn}^{(j)}}
\end{IEEEeqnarraybox}\right)
\label{eqn: jth step}
\end{align}
We call this process column elimination operation and continue performing until we achieve echelon form. Apparently, the number of steps equals rank of matrix $\Lambda$.
If any zero-column (column of all zeros) appears in any step, we shift this column to the right-hand side of matrix and apply the change to the transition matrix.

\subsection{Linear Block Codes}
Suppose $C$ denotes a binary linear block code $(n,k,d_{min})$. Then, any codeword $c \in C$ is a linear combination of the generator matrix rows of $C$. In fact, any codeword can be represented as a linear combination of the elements of a $k-$dimensional basis space.

The systematic generator matrix $G_{sys}$ of a linear block code $(n,k,d_{min})$ can be represented as two
sub-matrices: The $k \times k$ identity matrix, $I_{k}$, and a $k \times (n-k)$ parity sub-matrix $P$ such that
\begin{align}
{G_{sys}} = \left( {{I_k}\left| P \right.} \right),\\
{H_{sys}} =  - \left( {{P^T}\left| {{I_{n - k}}} \right.} \right).
\end{align}
$H_{sys}$, the systematic form of $H$, is an $(n-k)-$dimensional dual code of $C$; in other words, $H$ is a full rank matrix \cite{Morelos:02}.
The general form of generator matrix can be represented as two sub-matrices; The $k\times k$ matrix
of basis vectors, $B$, and a $k\times (n-k)$ parity sub-matrix $P$ such that
\begin{align}
G = \left( {B\left| P \right.} \right),\\
G{H^T} = \bar 0.
\end{align}
The basis sub-matrix $B$ can be represented as $B=\left( {{{B_1} \cdots {B_k}}} \right)$ where
$B_{j}$, $1 \le j \le k$, is the $j$th column of $B$, and the parity sub-matrix $P$ can
be shown as $P=\left( {{{P_1} \cdots {P_{n-k}}}} \right)$ , $1 \le j \le n-k$, is the $j$th column of $P$. Note that
\begin{align}
P_{j}=\left( {{{P_{1j}} \cdots {P_{kj}}}} \right)^T.
\end{align}
It is clear that $P_j$s can be rewritten as linear combinations of $B_j$s.
In fact, this relation between $P_j$ and $B_j$s is the same as the relation defined by parity check
equation of the code, i.e.
\begin{align}
P_j = \sum \limits_{m = 1}^k {p_{m,j}}{B_m},1 \le j \le n - k
\end{align}
\section{Parity Elimination on Linear Block Codes}
Since rank of $G$ is undoubtedly $k$, performing column elimination operation on $G$, will result in $n-k$ all-zero
columns in $G^{(k)}$. As mentioned earlier, column elimination operation is done in a manner that
this $n-k$ all-zero columns appear in the right-hand side of $G^{(k)}$. Therefore $G$ can be presented as
\begin{align}
{G^{(k)}} = \left( {{B_{k \times k}}\left| {\bar 0} \right.} \right),
\end{align}
where $B_{k \times k}$ is a basis set for $k-$dimensional vector space and $\bar 0$ is a $k\times(n-k)$ all-zero matrix.

\begin{theo} If column elimination operation is performed on the generator matrix, $G$, of a binary linear block code $C$, a basis set for  $(n-k)-$dimensional vector space will appear in the columns of transition matrix, $\Pi  = {\Gamma ^{(1)}} \times  \cdots  \times {\Gamma ^{(k)}}$, correspondent to the all-zero columns of $G^{(k)}$.
\end{theo}

\begin{proof} The procedure of column elimination operation is
\begin{align}
{G^{(k)}} = G \times {\Gamma ^{(1)}} \times  \cdots  \times {\Gamma ^{\left( k \right)}} = G \times \Pi  = \left( {{B_{k \times k}}\left| {{{\bar 0}_{k \times n - k}}} \right.} \right),
\label{eqn: Gtheorem}
\end{align}
On the other hand, any $n-k$ linear independent vectors, ${H^T} \ne \bar 0$, that satisfy $G{H^T} = \bar 0$, can form a basis set for dual code space. As shown in equation \eqref{eqn: jth step}, an identity matrix, $I_{n-k}$, appears in $G^{(k)}$, which means that $n-k$ vectors in the right-hand side of matrix $\Pi $ are linear independent. According to eq. \eqref{eqn: Gtheorem}, $n-k$ right-hand side columns of $\Pi $ multiplied by $G$ result in an all-zero sub-matrix in $G^{(k)}$, so they generate a basis for dual code space.
\end{proof}

\begin{exmp}
Let $G$ and $H$, the generator and parity check matrices of code $C$, Hamming code $(7,4)$, be,
\begin{align}
G = \left( {\begin{array}{*{20}{c}}
1&0&0&0&1&1&0\\
0&1&0&0&1&0&1\\
0&0&1&0&0&1&1\\
0&0&0&1&1&1&1
\end{array}} \right),
H = \left( {\begin{array}{*{20}{c}}
1&1&0&1&1&0&0\\
1&0&1&1&0&1&0\\
0&1&1&1&0&0&1
\end{array}} \right)
\end{align}
We perform column elimination operation for matrix $\Lambda$, its rows being codewords of $C$,
\begin{align}
{\rm{\Lambda }} = \left( {\begin{array}{*{20}{c}}
1&1&0&0&0&1&1\\
0&1&0&0&1&0&1\\
1&0&1&0&1&0&1\\
0&1&1&1&0&0&1
\end{array}} \right)
\end{align}
\begin{step}
\begin{equation}
{{\rm{\Lambda }}^{(1)}} = {\rm{\Lambda }} \times {\Gamma ^{(1)}} = \left( {\begin{array}{*{20}{c}}
1&0&0&0&0&0&0\\
0&1&0&0&1&0&1\\
1&1&1&0&1&1&0\\
0&1&1&1&0&0&1
\end{array}} \right),\Gamma ^{(1)}=
\left(\begin{IEEEeqnarraybox}[
\IEEEeqnarraystrutmode
][c]{c/v/c,c,c,c,c,c}
1&&1&0&0&0&1&1
\\\hline
0&&{}&{}&{}&{}&{}&{}\\
 \vdots &&{}&{}&{}&{{I_6}}&{}&{}\\
0&&{}&{}&{}&{}&{}&{}
\end{IEEEeqnarraybox}\right)
\end{equation}
\end{step}
\begin{step}
\begin{equation}
{{\rm{\Lambda }}^{(2)}} = {\Lambda }^{(1)} \times {\Gamma ^{(2)}} = \left( {\begin{array}{*{20}{c}}
1&0&0&0&0&0&0\\
0&1&0&0&0&0&0\\
1&1&1&0&0&1&1\\
0&1&1&1&1&0&0
\end{array}} \right),\Gamma ^{(2)}=
\left(\begin{IEEEeqnarraybox}[
\IEEEeqnarraystrutmode
][c]{c,c/v/c,c,c,c,c}
1&0&&0&0&0&0&0\\
0&1&&0&0&1&0&1
\\\hline
 0 &0&&{}&{}& &{}&{}\\
 \vdots & \vdots&&{}&{}&{{I_5}}&{}&{}\\
0&0&&{}&{}&{}&{}&{}
\end{IEEEeqnarraybox}\right)
\end{equation}
\end{step}
\begin{step}
\begin{equation}
{{\rm{\Lambda }}^{(3)}} = {\Lambda }^{(2)} \times {\Gamma ^{(3)}} = \left( {\begin{array}{*{20}{c}}
1&0&0&0&0&0&0\\
0&1&0&0&0&0&0\\
1&1&1&0&0&0&0\\
0&1&1&1&1&1&1
\end{array}} \right),\Gamma ^{(3)}=
\left(\begin{IEEEeqnarraybox}[
\IEEEeqnarraystrutmode
][c]{c,c,c/v/c,c,c,c}
 & & && &0 & 0 &0 \\
 & I_3 & &&  &0 &0 &0 \\
 & & && & 0& 1& 1\\\hline
 0 &0&0&&{}& &{}&{}\\
 \vdots & \vdots & \vdots &&{}& & I_4&{}\\
0&0&0&&{}&{}&{}&{}
\end{IEEEeqnarraybox}\right)
\end{equation}
\end{step}
\begin{step}
\begin{equation}
{{\rm{\Lambda }}^{(4)}} = {\Lambda }^{(3)} \times {\Gamma ^{(4)}} = \left( {\begin{array}{*{20}{c}}
1&0&0&0&0&0&0\\
0&1&0&0&0&0&0\\
1&1&1&0&0&0&0\\
0&1&1&1&0&0&0
\end{array}} \right),\Gamma ^{(4)}=
\left(\begin{IEEEeqnarraybox}[
\IEEEeqnarraystrutmode
][c]{c,c,c,c/v/c,c,c}
 & & & &&0 & 0 &0 \\
 & I_4 & &  &&0 &0 &0 \\
 & & & && 0& 0& 0\\
 & & & && 1& 1& 1\\\hline
 0 &0&0&0&& &{}&{}\\
 0 & 0 & 0 &0&& & I_3&{}\\
0&0&0&0&&{}&{}&{}
\end{IEEEeqnarraybox}\right)
\end{equation}
\end{step}
The transition matrix of the operation is as follows
\begin{equation}
\Pi  = {\Gamma ^{(1)}} \times {\Gamma ^{(2)}} \times {\Gamma ^{(3)}} \times {\Gamma ^{(4)}} =
\left(\begin{IEEEeqnarraybox}[
\IEEEeqnarraystrutmode
][c]{c,c,c,c/v/c,c,c}
1 & 1&0 & 0&&1 & 1 &0 \\
0 & 1 & 0&  0&&1 &0 &1 \\
 0& 0& 1& 0&& 0& 1& 1\\
 0&0 & 0& 1&& 1& 1& 1\\
 0 &0&0&0&& &{}&{}\\
 0 & 0 & 0 &0&& & I_3&{}\\
0&0&0&0&&{}&{}&{}
\end{IEEEeqnarraybox}\right)
\label{eqn: piexam1}
\end{equation}
Note that, as shown in equation \eqref{eqn: piexam1}, three right-hand columns of transition matrix, $\Pi$, form a basis set for dual code $3-$dimensional space.
\end{exmp}
\begin{exmp}
Consider code $C$ of example 1, and noisy matrix $\Lambda$ below,
\begin{align}
{\rm{\Lambda }} = \left( {\begin{array}{*{20}{c}}
1&0&0&1&\bf{1}&0&1\\
0&1&0&0&0&\bf{1}&0\\
\bf{1}&1&1&0&1&1&0\\
0&0&0&1&1&1&1
\end{array}} \right)
\end{align}
where noisy elements are denoted by boldface bits. Noisy elements are chosen arbitrarily
so as to satisfy ${\rm{\Lambda }} \times {h^T} = \bar 0$, $h = \left( 0\: 1\:1\:1\:0\:0\:1 \right)$.
We perform column elimination operation as follows:
\begin{step}
\begin{equation}
{{\rm{\Lambda }}^{(1)}} = {\rm{\Lambda }} \times {\Gamma ^{(1)}} = \left( {\begin{array}{*{20}{c}}
1&0&0&0&\bf{0}&0&0\\
0&1&0&1&0&\bf{1}&0\\
\bf{1}&1&1&\bf{1}&\bf{0}&1&\bf{1}\\
0&0&0&1&1&1&1
\end{array}} \right),\Gamma ^{(1)}=
\left(\begin{IEEEeqnarraybox}[
\IEEEeqnarraystrutmode
][c]{c/v/c,c,c,c,c,c}
1&&0&0&1&1&0&1
\\\hline
0&&{}&{}&{}&{}&{}&{}\\
 \vdots &&{}&{}&{}&{{I_6}}&{}&{}\\
0&&{}&{}&{}&{}&{}&{}
\end{IEEEeqnarraybox}\right)
\end{equation}
\end{step}
\begin{step}
\begin{equation}
{{\rm{\Lambda }}^{(2)}} = {\Lambda }^{(1)} \times {\Gamma ^{(2)}} = \left( {\begin{array}{*{20}{c}}
1&0&0&0&\bf{0}&0&0\\
0&1&0&0&0&\bf{0}&0\\
\bf{1}&1&1&\bf{0}&\bf{0}&\bf{0}&\bf{1}\\
0&0&0&1&1&1&1
\end{array}} \right),\Gamma ^{(2)}=
\left(\begin{IEEEeqnarraybox}[
\IEEEeqnarraystrutmode
][c]{c,c/v/c,c,c,c,c}
1&0&&0&0&0&0&0\\
0&1&&0&1&0&1&0
\\\hline
 0 &0&&{}&{}& &{}&{}\\
 \vdots & \vdots&&{}&{}&{{I_5}}&{}&{}\\
0&0&&{}&{}&{}&{}&{}
\end{IEEEeqnarraybox}\right)
\end{equation}
\end{step}
\begin{step}
\begin{equation}
{{\rm{\Lambda }}^{(3)}} = {\Lambda }^{(2)} \times {\Gamma ^{(3)}} = \left( {\begin{array}{*{20}{c}}
1&0&0&0&\bf{0}&0&0\\
0&1&0&0&0&\bf{0}&0\\
\bf{1}&1&1&\bf{0}&\bf{0}&\bf{0}&0\\
0&0&0&1&1&1&1
\end{array}} \right),\Gamma ^{(3)}=
\left(\begin{IEEEeqnarraybox}[
\IEEEeqnarraystrutmode
][c]{c,c,c/v/c,c,c,c}
 & & && &0 & 0 &0 \\
 & I_3 & &&  &0 &0 &0 \\
 & & && & 0& 0& 1\\\hline
 0 &0&0&&{}& &{}&{}\\
 \vdots & \vdots & \vdots &&{}& & I_4&{}\\
0&0&0&&{}&{}&{}&{}
\end{IEEEeqnarraybox}\right)
\end{equation}
\end{step}
\begin{step}
\begin{equation}
{{\rm{\Lambda }}^{(4)}} = {\Lambda }^{(3)} \times {\Gamma ^{(4)}} = \left( {\begin{array}{*{20}{c}}
1&0&0&0&\bf{0}&0&0\\
0&1&0&0&0&\bf{0}&0\\
\bf{1}&1&1&\bf{0}&\bf{0}&\bf{0}&0\\
0&0&0&1&0&0&0
\end{array}} \right),\Gamma ^{(4)}=
\left(\begin{IEEEeqnarraybox}[
\IEEEeqnarraystrutmode
][c]{c,c,c,c/v/c,c,c}
 & & & &&0 & 0 &0 \\
 & I_4 & &  &&0 &0 &0 \\
 & & & && 0& 0& 0\\
 & & & && 1& 1& 1\\\hline
 0 &0&0&0&& &{}&{}\\
 0 & 0 & 0 &0&& & I_3&{}\\
0&0&0&0&&{}&{}&{}
\end{IEEEeqnarraybox}\right)
\end{equation}
\end{step}
The transition matrix of the operation on the noisy data is as follows
\begin{equation}
\Pi  = {\Gamma ^{(1)}} \times {\Gamma ^{(2)}} \times {\Gamma ^{(3)}} \times {\Gamma ^{(4)}} =
\left(\begin{IEEEeqnarraybox}[
\IEEEeqnarraystrutmode
][c]{c,c,c,c/v/c,c,c}
{1} & 0&0 & 1&&0 & 1 & {0} \\
{0} & 1 & 0&  1&&1 &0 & {1} \\
 0& 0& 1& 0&& 0& 0&{1}\\
 0&0 & 0& 1&& 1& 1& {1}\\
 0 &0&0&0&& 1 & 0 &{0}\\
 0 & 0 & 0 &0&& 0 & 1 &{0}\\
0&0&0&0&& 0 & 0 &{1}
\end{IEEEeqnarraybox}\right)
\end{equation}
\end{exmp}
The right-hand column of $\Pi$ is a dual codeword (i.e. $h$). Since, parity check equation of vector
$h = \left( 0\: 1\:1\:1\:0\:0\:1 \right)$ expresses that sum of second, third, fourth and seventh columns of any
matrix consisting of codewords (e.g. $G$) equals zero, but errors appeared in the first and fifth columns,
which are not included in the parity check equation related to this vector.

Note that there will certainly be three all-zero columns after column elimination operation,
due to the rank of matrix being limited to the number of rows (and not columns). Therefore,
all the columns correspondent to all-zero columns of transition matrix cannot be deemed as dual
codewords. To find out whether columns corresponding to all-zero columns of transition matrix are
dual codewords or not, more codewords are required. This will be discussed in subsection IV.

Although the error rate is very high, around $0.07$, one of the column of $H^T$ was recovered.
This leads to expectancy of recovering dual codewords in relatively large error rates.

\begin{resu}
In order to obtain any basis vector of dual code space (any column of matrix $H^T$),
only $k$ independent (valid or invalid) codewords $\upsilon_j,\,\, 1 \le j \le k$ are needed
such that they satisfy ${\upsilon _j}{h^T} = \bar 0,\,\, 1 \le j \le k$.
\end{resu}
\section{Recovering Dual Codewords from Noisy Data}
Suppose that the transmission channel is a binary symmetric channel (BSC) with crossover probability $\varepsilon  \ll \frac{1}{2}$. Suppose, also that the message is encoded using a linear block code $C(n,k)$.
Intercepted bitstream from the channel is first divided into words of length $n$, and matrix $\Lambda$ is constructed such that each word is placed in one row.
Note that the number of rows of matrix $\Lambda$ should be chosen in a way that rank of $\Lambda$  is
not constrained by the number of rows (the number of codewords required for recovery operation is
discussed in section V). Now, we introduce the iterative elimination algorithm to recover the parity check matrix.
For $\varepsilon  \ll \frac{1}{2}$, parity check matrix is achievable in low iterations.

\subsection*{Iterative Elimination Algorithm}
\begin{iter}
Consider a $k \times n$ window $W_1$. We place the window on the first $k$ rows of matrix ${{\rm{\Lambda }}^{(0,1)}}$ and perform the column elimination operation such that echelon form is achieved, i.e.
\begin{equation}
  {{\rm{\Lambda }}^{(0,1)}}={\rm{\Lambda }} =
    \left(
    \begin{array}{cccccc}\hline
\multicolumn{1}{|c}{{a_{11}^{(0,0)}}} & \cdots &{a_{1k}^{(0,0)}} & {a_{1(k + 1)}^{(0,0)}}& \cdots &\multicolumn{1}{c|}{a_{1n}^{(0,0)}}\\
\multicolumn{1}{|c}{\vdots} & \ddots & \vdots & \vdots & \ddots & \multicolumn{1}{c|}{\vdots} \\
\multicolumn{1}{|c}{a_{k1}^{(0,0)}} & \cdots &{a_{kk}^{(0,0)}}&{a_{k(k + 1)}^{(0,0)}}& \cdots &\multicolumn{1}{c|}{a_{kn}^{(0,0)}}\\\hline
{a_{(k + 1)1}^{(0,0)}}& \cdots &{a_{(k + 1)k}^{(0,0)}}&{a_{(k + 1)(k + 1)}^{(0,0)}}& \cdots &{a_{(k + 1)n}^{(0,0)}}\\
 \vdots & \ddots & \vdots & \vdots & \ddots & \vdots \\
{a_{M1}^{(0,0)}}& \cdots &{a_{Mk}^{(0,0)}}&{a_{M(k + 1)}^{(0,0)}}& \cdots &{a_{Mn}^{(0,0)}}
    \end{array}
    \right)
\end{equation}
Superscripts $(.,.)$ mean the number of column elimination step in this iteration and window number respectively. According to result $1$, formation of any all-zero columns in matrix $\Lambda ^{(k,1)}$ means that the correspondent column in the transition matrix is a dual codeword. But due to presence of noise (noise increases the rank of matrix), it is very likely not to have $n-k$ all-zero columns in the right-hand side of matrix $\Lambda ^{(k,1)}$, even in case that dual codewords appear in the transition matrix. In other words, in case of noisy channel, the probability of having all-zero syndrome matrix $s$ (the quantity $s = {\rm{\Lambda }}{H^T}$ is called the syndrome matrix and it is known at the receiver) is very low.

Knowing that $\varepsilon  \ll \frac{1}{2}$ and from result $1$, if we can find $k$ independent rows $\upsilon _j$ in the window such that ${\upsilon _j}{h^T} = \bar 0,\,\, 1 \le j \le k$ ($h^T$ is a column of $H^T$), then a low-weight  column (for which a criterion will be calculated in section V) will appear in $\Lambda ^{(k,1)}$, and the corresponding column in transition matrix $\Pi$ can be admitted (with some probability of false alarm) as a basis for dual code. Therefore, in case of noisy data, we attempt to build low-weigh columns in matrix $\Lambda ^{(k,{\kern 1pt} .)}$.
\end{iter}
\begin{iter}
$\Lambda ^{(0,2)} = \Lambda ^{(k,0)}$. Now we slide the window down by $k$ rows ($W_2$), and similar to the first iteration, by performing column elimination operation on $\Lambda ^{(0,2)}$,we reduce $W_2$ to echelon form i.e.
\begin{equation}
\Lambda ^{(0,2)} =
    \left(
    \begin{array}{cccccc}
{a_{11}^{\left( {0,1} \right)}}& \cdots &{a_{1k}^{\left( {0,1} \right)}}&{a_{1\left( {k + 1} \right)}^{\left( {0,1} \right)}}& \cdots &{a_{1n}^{\left( {0,1} \right)}}\\\hline
\multicolumn{1}{|c}{a_{21}^{(0,1)}}& \cdots &{a_{2k}^{(0,1)}}&{a_{2(k + 1)}^{(0,1)}}& \cdots & \multicolumn{1}{c|}{a_{2n}^{(0,1)}}\\
 \multicolumn{1}{|c}{\vdots} & \ddots & \vdots & \vdots & \ddots & \multicolumn{1}{c|}{\vdots} \\
\multicolumn{1}{|c}{a_{(k + 1)1}^{(0,1)}}& \cdots &{a_{(k + 1)k}^{(0,1)}}&{a_{(k + 1)(k + 1)}^{(0,1)}}& \cdots & \multicolumn{1}{c|}{a_{(k + 1)n}^{(0,1)}}\\\hline
{a_{(k + 2)1}^{(0,1)}}& \cdots &{a_{(k + 2)k}^{(0,1)}}&{a_{(k + 2)(k + 2)}^{(0,1)}}& \cdots &{a_{(k + 2)n}^{(0,1)}}\\
 \vdots & \ddots & \vdots & \vdots & \ddots & \vdots \\
{a_{M1}^{(0,1)}}& \cdots &{a_{Mk}^{(0,1)}}&{a_{M(k + 1)}^{(0,1)}}& \cdots &{a_{Mn}^{(0,1)}}
    \end{array}
    \right)
\end{equation}
Again we look for low-weigh columns in $\Lambda ^{(k,2)}$ and correspondent columns in $\Pi^{(2)}$.
The iterations are continued until we obtain a parity check matrix.

Notice that, at each iteration, the matrix obtained in the previous iteration is used.We will show in the next theorem that not only invalid parity equations of each iteration does not affect the next iteration, but also valid parity equations pass to the next iteration.The iterative elimination algorithm is demonstrated in table I.

As shown in \cite{Valembois2001199} and \cite{Cluzeau:isit09}, the weight of $\Lambda  \times {h^T}$ is dependent on the channel's error; so, we use the weight of the columns of $\Lambda ^{(k,i)},\;1 \le i \le \frac{M}{k}$ to detect the words of $H$. If a column in $\Pi^{(i)}$ belongs to the dual code words, the weight of its correspondent column in $\Lambda ^{(k,i)}$ will be much smaller than $\frac{M}{2}$ since $\varepsilon  \ll \frac{1}{2}$. This is the reason of searching for low-weigh columns.
\end{iter}

\begin{theo}
In iterative elimination algorithm, error in window $W_{i-1}$ ($(i-1)$th iteration) does not pass to the next iterations and if any window consists of $k$ independent valid codewords, then basis of dual code space will appear in transition matrix $\Pi^{(i)}$ .
\end{theo}
\begin{proof}
Matrix $\Lambda$, in the transmitter, can be shown as follows
\begin{align}
{{\rm{\Lambda }}^{(correct)}} = \left ( {\left. B \right|{P^{(correct)}}} \right)
\end{align}
where $B$ is a basis set for the vector space that columns of matrix ${{\rm{\Lambda }}^{(correct)}}$ are its elements, and $P^{(correct)}$ is the sub-matrix of dependent columns. We show the columns of these two matrices as $b_j,\:1\leq j\leq k$ and $p_j^{(correct)},\:1\leq j\leq k$ respectively. We have the following relation between $b_j$s and $p_j^{(correct)}$s
\begin{align}
p_j^{(correct)} = \mathop \sum \limits_{m = 1}^k {\alpha _{\left( {j,m} \right)}}{b_m}
\label{eqn: Pcorrect}
\end{align}
where ${\alpha _{\left( {j,m} \right)}} \in \left\{ {0,1} \right\}$. ${\alpha _{\left( {j,m} \right)}} = 1$ If $P^{(correct)}$ is dependent on $b_m$, else, ${\alpha _{\left( {j,m} \right)}}= 0$.
Clearly, permutation of rows does not change ${\alpha _{\left( {j,m} \right)}}$s. Therefore, columns of the
transformed matrix satisfy equation \eqref{eqn: Pcorrect}. Note that, in this case, elements of $b_j$s and $P^{(correct)}$s permute.
\begin{exmp}
Assume that matrix ${{\rm{\Lambda }}^{(correct)}}$ (with rank $2$) is as follows
\begin{equation}
{{\rm{\Lambda }}^{(correct)}} =
\left(\begin{IEEEeqnarraybox}[
\IEEEeqnarraystrutmode
][c]{c,c/v/c,c,c}
1&0&&1&1&0\\
0&1&&1&0&1\\
1&1&&0&1&1
\end{IEEEeqnarraybox}\right)
\end{equation}
Apparently, $P_j^{(correct)}$ is a combination of $b_1$ and $b_2$
\begin{equation}
p_1^{(correct)} =
\left(\begin{IEEEeqnarraybox}[
\IEEEeqnarraystrutmode
][c]{c}
1\\
1\\
0
\end{IEEEeqnarraybox}\right)=
\left(\begin{IEEEeqnarraybox}[
\IEEEeqnarraystrutmode
][c]{c}
1\\
0\\
1
\end{IEEEeqnarraybox}\right)+
\left(\begin{IEEEeqnarraybox}[
\IEEEeqnarraystrutmode
][c]{c}
0\\
1\\
1
\end{IEEEeqnarraybox}\right)
\end{equation}
Consequently,
\begin{align}
p_1^{(correct)} = \mathop \sum \limits_{m = 1}^2 {\alpha _{\left( {1,m} \right)}}{b_m},\quad  \Rightarrow \quad \left\{ {\begin{array}{*{20}{c}}
{{\alpha _{(1,1)}} = 1}\\
{\alpha _{(1,2)} = 1}
\end{array}} \right.
\end{align}
\begin{align}
p_2^{(correct)} = \mathop \sum \limits_{m = 1}^2 {\alpha _{\left( {2,m} \right)}}{b_m},\quad  \Rightarrow \quad \left\{ {\begin{array}{*{20}{c}}
{{\alpha _{(2,1)}} = 1}\\
{\alpha _{(2,2)} = 0}
\end{array}} \right.
\end{align}
\begin{align}
p_3^{(correct)} = \mathop \sum \limits_{m = 1}^2 {\alpha _{\left( {3,m} \right)}}{b_m},\quad  \Rightarrow \quad \left\{ {\begin{array}{*{20}{c}}
{{\alpha _{(3,1)}} = 1}\\
{\alpha _{(3,2)} = 1}
\end{array}} \right.
\end{align}
\end{exmp}
Now, consider window $W_1$ in noisy data matrix $\Lambda ^{(0,1)}$ (iteration 1). Before performing the operations, for each column in windows $W_1$, $W_2$ and $W_i$ we can write
\begin{align}
p_j^{(0,1,1)} = \mathop \sum \limits_{m = 1}^k \left( {{\alpha _{\left( {j,m} \right)}} + \beta _{\left( {j,m} \right)}^1} \right)b_m^1,\quad 1 \le j \le n - k
\end{align}
\begin{align}
p_j^{(0,1,2)} = \mathop \sum \limits_{m = 1}^k \left( {{\alpha _{\left( {j,m} \right)}} + \beta _{\left( {j,m} \right)}^2} \right)b_m^2,\quad 1 \le j \le n - k
\end{align}
\begin{align}
p_j^{(0,1,i)} = \mathop \sum \limits_{m = 1}^k \left( {{\alpha _{\left( {j,m} \right)}} + \beta _{\left( {j,m} \right)}^i} \right)b_m^i,\quad 1 \le j \le n - k
\end{align}
where term $\mathop \sum \limits_{m = 1}^k \beta _{\left( {j,m} \right)}^ib_m^i,\: \beta _{\left( {j,m} \right)}^i \in \left\{ {0,1} \right\}$ is added to eq. \eqref{eqn: Pcorrect} due to channel error. In the proof, superscripts $(.,.,.)$ denote the number of column elimination step, iteration and window respectively. So $(0,1,i)$ means the $i$th window ($i$th $k$ rows), before performing column elimination of $1$th iteration, and $(k,1,i)$ means the $i$th window, after performing column elimination of $1$th iteration. Remember that in iterative elimination algorithm, operations are performed on the whole rows of matrix, regardless of being in the window or not.

After performing iteration 1 of iterative elimination algorithm, consider window $W_1$, and implicitly windows $W_2$ and $W_i$. The following relations can be written for columns in each window
\begin{align}
p_j^{(k,1,1)} = \mathop \sum \limits_{m = 1}^k \left( {{\alpha _{\left( {j,m} \right)}} + \beta _{\left( {j,m} \right)}^1} \right)b_m^1 + \mathop \sum \limits_{m = 1}^k \left( {{\alpha _{\left( {j,m} \right)}} + \beta _{\left( {j,m} \right)}^1} \right)b_m^1 = \bar 0,\quad 1 \le j \le n - k
\end{align}
\begin{align}
p_j^{(k,1,2)} = \mathop \sum \limits_{m = 1}^k \left( {{\alpha _{\left( {j,m} \right)}} + \beta _{\left( {j,m} \right)}^2} \right)b_m^2 + \mathop \sum \limits_{m = 1}^k \left( {{\alpha _{\left( {j,m} \right)}} + \beta _{\left( {j,m} \right)}^1} \right)b_m^2 = \mathop \sum \limits_{m = 1}^k \left( {\beta _{\left( {j,m} \right)}^1 + \beta _{\left( {j,m} \right)}^2} \right)b_m^2,
\end{align}
\begin{align}
p_j^{(k,1,i)} = \mathop \sum \limits_{m = 1}^k \left( {{\alpha _{\left( {j,m} \right)}} + \beta _{\left( {j,m} \right)}^i} \right)b_m^i + \mathop \sum \limits_{m = 1}^k \left( {{\alpha _{\left( {j,m} \right)}} + \beta _{\left( {j,m} \right)}^1} \right)b_m^i = \mathop \sum \limits_{m = 1}^k \left( {\beta _{\left( {j,m} \right)}^1 + \beta _{\left( {j,m} \right)}^i} \right)b_m^i,
\end{align}
we get transition matrix ${\Pi ^{(1)}}$ as follows
\begin{equation}
{\Pi ^{(1)}} =
\left(\begin{IEEEeqnarraybox}[
\IEEEeqnarraystrutmode
][c]{c,c,c/v/c,c,c}
{}&{}&{}&&{{\alpha _{\left( {1,1} \right)}} + \beta _{\left( {1,1} \right)}^1}& \cdots &{{\alpha _{\left( {n - k,1} \right)}} + \beta _{\left( {n - k,1} \right)}^1}\\
{}&{{I_n}}&{}&& \vdots & \ddots & \vdots \\
{}&{}&{}&&{{\alpha _{\left( {1,k} \right)}} + \beta _{\left( {1,k} \right)}^1}& \cdots &{{\alpha _{\left( {n - k,k} \right)}} + \beta _{\left( {n - k,k} \right)}^1}\\\hline
0& \cdots &0&&{}&{}&{}\\
 \vdots & \ddots & \vdots &&{}&{{I_{n - k}}}&{}\\
0& \cdots &0&&{}&{}&{}
\end{IEEEeqnarraybox}\right)
\end{equation}
Similarly, in iteration 2, we can write for windows $W_2$ and implicitly $W_i$ before operation
\begin{align}
p_j^{(0,2,2)} = \mathop \sum \limits_{m = 1}^k \left( {\beta _{\left( {j,m} \right)}^1 + \beta _{\left( {j,m} \right)}^2} \right)b_m^2,\quad 1 \le j \le n - k
\end{align}
\begin{align}
p_j^{(0,2,i)} = \mathop \sum \limits_{m = 1}^k \left( {\beta _{\left( {j,m} \right)}^1 + \beta _{\left( {j,m} \right)}^i} \right)b_m^i,\quad 1 \le j \le n - k
\end{align}
and after operation
\begin{align}
p_j^{(k,2,2)} = \mathop \sum \limits_{m = 1}^k \left( {\beta _{\left( {j,m} \right)}^1 + \beta _{\left( {j,m} \right)}^2} \right)b_m^2 + \mathop \sum \limits_{m = 1}^k \left( {\beta _{\left( {j,m} \right)}^1 + \beta _{\left( {j,m} \right)}^2} \right)b_m^2 = \bar 0,\quad 1 \le j \le n - k
\end{align}
\begin{align}
p_j^{(k,1,i)} = \mathop \sum \limits_{m = 1}^k \left( {\beta _{\left( {j,m} \right)}^1 + \beta _{\left( {j,m} \right)}^i} \right)b_m^i + \mathop \sum \limits_{m = 1}^k \left( {\beta _{\left( {j,m} \right)}^1 + \beta _{\left( {j,m} \right)}^2} \right)b_m^i = \mathop \sum \limits_{m = 1}^k \left( {\beta _{\left( {j,m} \right)}^2 + \beta _{\left( {j,m} \right)}^i} \right)b_m^i,
\end{align}
the resultant transition matrix of iteration 2 is
\begin{equation}
{\Pi ^{(2)}} =
\left(\begin{IEEEeqnarraybox}[
\IEEEeqnarraystrutmode
][c]{c,c,c/v/c,c,c}
{}&{}&{}&&{{\alpha _{\left( {1,1} \right)}} + \beta _{\left( {1,1} \right)}^2}& \cdots &{{\alpha _{\left( {n - k,1} \right)}} + \beta _{\left( {n - k,1} \right)}^2}\\
{}&{{I_n}}&{}&& \vdots & \ddots & \vdots \\
{}&{}&{}&&{{\alpha _{\left( {1,k} \right)}} + \beta _{\left( {1,k} \right)}^2}& \cdots &{{\alpha _{\left( {n - k,k} \right)}} + \beta _{\left( {n - k,k} \right)}^2}\\\hline
0& \cdots &0&&{}&{}&{}\\
 \vdots & \ddots & \vdots &&{}&{{I_{n - k}}}&{}\\
0& \cdots &0&&{}&{}&{}
\end{IEEEeqnarraybox}\right)
\end{equation}
thus, error related to $W_1$ does not appear in the transition matrix $\Pi ^{(2)}$.Consider iteration $i$ and window $W_i$ before operations
\begin{align}
p_j^{(0,i,i)} = \mathop \sum \limits_{m = 1}^k \left( {\beta _{\left( {j,m} \right)}^{i - 1} + \beta _{\left( {j,m} \right)}^i} \right)b_m^i,\quad 1 \le j \le n - k
\end{align}
after operations of iteration $i$, we have the following
\begin{align}
p_j^{(k,i,i)} = \mathop \sum \limits_{m = 1}^k \left( {\beta _{\left( {j,m} \right)}^{i - 1} + \beta _{\left( {j,m} \right)}^i} \right)b_m^i + \mathop \sum \limits_{m = 1}^k \left( {\beta _{\left( {j,m} \right)}^{i - 1} + \beta _{\left( {j,m} \right)}^i} \right)b_m^i = \bar 0,\quad 1 \le j \le n - k
\end{align}
if the $i$th window consists of $k$ independent valid codewords ${\upsilon _j},\;1 \le j \le k$ i.e. $\beta _{\left( {j,m} \right)}^i = 0,\;1 \le m \le k,\;1 \le j \le k$ then we have
\begin{equation}
{\Pi ^{(i)}} =
\left(\begin{IEEEeqnarraybox}[
\IEEEeqnarraystrutmode
][c]{c,c,c/v/c,c,c}
{}&{}&{}&&{\alpha _{\left( {1,1} \right)}} + 0& \cdots &{{\alpha _{\left( {n - k,1} \right)}} + 0}\\
{}&{{I_n}}&{}&& \vdots & \ddots & \vdots \\
{}&{}&{}&&{{\alpha _{\left( {1,k} \right)}} + 0}& \cdots &{{\alpha _{\left( {n - k,k} \right)}} + 0}\\\hline
0& \cdots &0&&{}&{}&{}\\
 \vdots & \ddots & \vdots &&{}&{{I_{n - k}}}&{}\\
0& \cdots &0&&{}&{}&{}
\end{IEEEeqnarraybox}\right)
\end{equation}
where columns in the right-hand side area basis set for dual code space.
\end{proof}
To recover any word of the parity check matrix, $h$, we need $k$ independent
codewords ${\upsilon _j},\;1 \le j \le k$ at the $i$th iteration such that they satisfy ${\upsilon _j}{h^T} = 0$, for $1 \le j \le k$.$h^T$ appears in the transition matrix at the $i$th iteration but because of the error in the data
matrix , $\mathop \sum \limits_{m = 1}^M {\rm{\Lambda }}_{mj}^{(k,i)}$  doesn't have zero weight.
\section{Computing the Threshold T}
If $h$ belongs to the dual code of $C$, ${\rm{\Lambda }} \times {h^T}$ will have a Gaussian distribution with mean $\frac{M}{2}\left( {1 - {{(1 - 2\varepsilon )}^{w{t_H}(h)}}} \right)$ and variance$\frac{M}{4}\left( {1 - {{(1 - 2\varepsilon )}^{2w{t_H}(h)}}} \right)$; Otherwise, ${{\rm{\Lambda }}_M}{h^T}$ will have a Gaussian distribution with mean $\frac{M}{2}$ and variance $\frac{M}{4}$ \cite{Cluzeau:isit09}.

If $Z_j$ is the weight of $i$th column and ${\rm{\Lambda }}_{mj}^{(k,i)}$ is the $(m,j)$ element of ${{\rm{\Lambda }}^{(k,i)}}$, we will have
\begin{align}
{Z_j} = \mathop \sum \limits_{m = 1}^M {\rm{\Lambda }}_{mj}^{(k,i)},
\end{align}
According to the parity elimination algorithm there are two alternatives:

\:${{{\cal H}_0}}$: If in the algorithm procedure the dual codeword appears in the $j$th column of the transition matrix, $Z_j$ will have a Gaussian distribution with mean and variance
\begin{align}
{m_{{Z_j}}} = \frac{M}{2}\left( {1 - {{(1 - 2\varepsilon )}^{w{t_H}(h)}}} \right),\\
\sigma _{{Z_j}}^2 = \frac{M}{4}\left( {1 - {{(1 - 2\varepsilon )}^{2w{t_H}(h)}}} \right),
\end{align}
respectively.

\:${{{\cal H}_1}}$: Otherwise, $Z_j$ will have a Gaussian distribution with mean and variance
\begin{align}
{m_{{Z_j}}} = \frac{M}{2},\\
\sigma _{{Z_j}}^2 = \frac{M}{4},
\end{align}
False alarm probability $(p_{fa})$: is the probability of $Z_j$ being smaller than $T$, but its correspondent column
in the transition matrix does not appear in the parity check matrix.
\begin{align}
{p_{fa}} &= P\left( {{Z_j} < T\left| {{{\cal H}_1}} \right.} \right)\\
{p_{fa}} &= \int_{ - \infty }^T {\frac{1}{{\sqrt {2\pi \left( {{M \mathord{\left/
 {\vphantom {M 4}} \right.
 \kern-\nulldelimiterspace} 4}} \right)} }}{e^{ - \frac{{{{\left( {x - \left( {{M \mathord{\left/
 {\vphantom {M 2}} \right.
 \kern-\nulldelimiterspace} 2}} \right)} \right)}^2}}}{{2\left( {{M \mathord{\left/
 {\vphantom {M 4}} \right.
 \kern-\nulldelimiterspace} 4}} \right)}}}}dx}\nonumber\\
 &= \phi \left( {\frac{{T - {M \mathord{\left/
 {\vphantom {M 2}} \right.
 \kern-\nulldelimiterspace} 2}}}{{\sqrt {{M \mathord{\left/
 {\vphantom {M 4}} \right.
 \kern-\nulldelimiterspace} 4}} }}} \right)
\label{eqn: p_fa}
\end{align}
Non-detection probability $(p_{nd})$: is the probability of $Z_j$ being greater than $T$, but its correspondent
column in the transition matrix appears in the Parity check matrix.
\begin{align}
{p_{nd}} &= P\left( {{Z_j} > T\left| {{{\cal H}_0}} \right.} \right),\\
{p_{nd}} &= \int_T^{ + \infty } {\frac{1}{{\sqrt {2\pi \left( {M\alpha \left( {1 - \alpha } \right)} \right)} }}{e^{ - \frac{{{{\left( {x - \left( {M\alpha } \right)} \right)}^2}}}{{2\left( {M\alpha \left( {1 - \alpha } \right)} \right)}}}}dx}\nonumber\\
&= 1 - \phi \left( {\frac{{T - M\alpha }}{{\sqrt {M\alpha \left( {1 - \alpha } \right)} }}} \right),
\label{eqn: p_nd}
\end{align}
Where $\alpha  = \frac{{\left( {1 - {{(1 - 2\varepsilon )}^{w{t_H}(h)}}} \right)}}{2}$.From \eqref{eqn: p_fa} and \eqref{eqn: p_nd} we have
\begin{align}
M = {\left( {\frac{{{\phi ^{ - 1}}\left( {1 - {p_{nd}}} \right)\sqrt {1 - {{\left( {1 - 2\varepsilon } \right)}^{2w{t_H}(h)}}}  - {\phi ^{ - 1}}\left( {{p_{fa}}} \right)}}{{1 - {{\left( {1 - 2\varepsilon } \right)}^{2w{t_H}(h)}}}}} \right)^2},
\label{eqn: M}
\end{align}
\begin{align}
T = \frac{1}{2}\left( {M + {\phi ^{ - 1}}\left( {{p_{fa}}} \right)\sqrt M } \right),
\label{eqn: T}
\end{align}
Block codes and especially LDPC codes have words with very low weight in their parity check matrix. Thus,
$M$ does not significantly change even with increase of $\varepsilon $ \cite{chabot:isit07}. Therefore, we choose
the greatest $w{t_H}\left( h \right)$.
\section{Practical Experiments}
In the previous section we have shown how to find words of parity check matrix.
Table II gives some results of performing parity elimination algorithm on random linear codes of rate
$\frac{1}{2}$.

In comparison to \cite{Cluzeau:isit09}, results obtained by iterative elimination algorithm are better. Decrease in the number of recovered dual codewords with increase in code length is due to the increase in error probability in each codeword.

Complexity of ${\cal O}\left( {{n^3}} \right)$ and memory required of order ${\cal O}\left( {{n^2}} \right)$ makes this algorithm practical.
Furthermore, the structure of algorithm enables us to perform it on long length codes, by exploitation of parallel processors.
On the other hand, there is no need to a priori knowledge about the Hamming weigh of parity check matrix.
\section{Conclusion}
In this paper we introduced a method with very low complexity to recover the parity check matrix of a binary linear block code. Due to using only XOR operator, the algorithm can be easily implemented. A very important contribution in this algorithm is that (unlike \cite{Cluzeau:isit09}) there is no need to search on Hamming weigh of dual codewords, i.e. Parity elimination algorithm can recover dual codewords of different weights simultaneously.

As mentioned, increase in code length might result in more average number of bit errors in each codeword and therefore more iterations are required in the algorithm. But, there might still be enough valid codewords in the noisy data matrix, though not placed in the same window. Instead of searching for dual codewords, one can recover valid codewords. In the next paper, we will introduce a method to detect errors in the noisy data matrix.


\newpage
\begin{table}[hb]
\label{table: elimination algorithm} \caption{Parity Elimination Algorithm}
\begin{center}
\begin{tabular}{|l|}
\hline
Choose $M > eq.(59)$ \\
Compute $T$ with $eq.(60)$ \\
Choose number of Iteration $I$  \\
$for\;  i=1: I$\\
\;\;\;\;\;\; $\Lambda ^{(k,i + 1)} = \Lambda^{(0,i)}\, \Pi ^{(i + 1)}$\\
\;\;\;\;\;\; $for\;  j=k+1: n$\\
\;\;\;\;\;\;\;\;\;\;\;\; ${Z_j} = \mathop \sum \limits_{m = 1}^M {\rm{\Lambda }}_{mj}^{(k,i)}$\\
\;\;\;\;\;\;\;\;\;\;\;\;  $if \; {Z_j} < T$\\
\;\;\;\;\;\;\;\;\;\;\;\; Store $j$th column of ${\Pi ^{(i + 1)}}$\\
\;\;\;\;\;\; $end$ \\
$end$ \\
\hline
\end{tabular}
\end{center}
\end{table}

\begin{table}[hb]
\label{tab:recoword} \caption{Average number of parity check words recovered in $10000$ run of the parity elimination algorithm on random code of rate $\frac{1}{2}$. $n$  is the length of codewords and $\varepsilon$ is crossover probability and $M=10k$, $I=10$.}
\begin{center}
\begin{tabular}{|c|c|c|c|c|c|c|}
  \hline
  $\varepsilon$ & 0.001	& 0.002	& 0.005 & 0.01	& 0.02 & 0.05\\
  $n$ & & & & & & \\\hline
  32 & 15.9 & 10.5 & 8.7 & 5.6 & 3.2 & 1.4\\\hline
  64 & 31.7 & 26.9 & 18.6 & 8.6 & 1.7 & 0.4\\\hline
  128 & 63.4 & 51.1 & 27.3 & 1.2 & 0 & 0 \\\hline
  256 & 127.2 & 83.8 & 0 & 0 & 0 & 0 \\\hline
\end{tabular}
\end{center}
\end{table}

\newpage

\bibliographystyle{vancouver}
\bibliography{IETreference}

\end{spacing}

\end{document}